\pdfoutput=1
\documentclass[12pt,a4paper]{amsart}
\usepackage[pdftex]{graphicx}
\usepackage{amssymb}
\usepackage{amsmath}
\usepackage[numbers]{natbib}
\usepackage{url}
\usepackage{lscape}

\usepackage{algorithm}
\usepackage{algorithmic}

\newtheorem{defi}{Definition}[section]
\newtheorem{prop}{Proposition}
\newtheorem{cor}{Corollary}

\newcommand{\cell}[2]{\left(#1,#2\right)}
\newcommand{\xvar}[3]{X\left(#1,#2,#3\right)}
\newcommand{\yvar}[4]{Y\left(#1,#2,#3,#4\right)}

\newcommand{\zvar}[1]{Z\left(#1\right)}

\newcommand{\assign}[3]{#1\left(#2,#3\right)}
\newcommand{\uvar}[2]{U\left(#1,#2\right)}

\title{Exact Method for Generating Strategy-Solvable Sudoku Clues}

\author[K.\ Nishikawa]{Kohei Nishikawa}
\address[K.\ Nishikawa]{Graduate School of Information Systems, University of Electro-Communications, 1-5-1 Chofugaoka, Chofu, Tokyo 182-8585, Japan.}

\author[T.\ Toda]{Takahisa Toda}
\address[T.\ Toda]{Graduate School of Information Systems, University of Electro-Communications, 1-5-1 Chofugaoka, Chofu, Tokyo 182-8585, Japan.}
\email[Corresponding author]{todat@acm.org}

\keywords{Sudoku; Constraint Satisfaction Problem; Strategy-Solvability; Exact Method; Mathematics of Sudoku; Constraints of Clue Geometry; Strategy-solvable minimum Sudoku}

\begin{document}

\maketitle

\begin{abstract}
A Sudoku puzzle often has a regular pattern in the arrangement of initial digits and it is typically made solvable with known solving techniques, called strategies.
In this paper, we consider the problem of generating such Sudoku instances.
We introduce a rigorous framework to discuss solvability for Sudoku instances with respect to strategies.
This allows us to handle not only known strategies but also general strategies under a few reasonable assumptions.
We propose an exact method for determining Sudoku clues for a given set of clue positions that is solvable with a given set of strategies.
This is the first exact method except for a trivial brute-force search.
Besides the clue generation, we present an application of our method to the problem of determining the minimum number of strategy-solvable Sudoku clues.
We conduct experiments to evaluate our method, varying the position and the number of clues at random.
Our method terminates within $1$ minutes for many grids.
However, as the number of clues gets closer to $20$, the running time rapidly increases and exceeds the time limit set to $600$ seconds.
We also evaluate our method for several instances with $17$ clue positions taken from known minimum Sudokus to see the efficiency for deciding unsolvability.
\end{abstract}

\section{Introduction}\label{sec:introduction}
Sudoku is a popular number-placement puzzle.
In an ordinary Sudoku (Figure~\ref{fig:sudoku}), given a partially completed $9\times 9$ grid, the goal is to fill in all empty cells with digits from $1$ to $9$ in such a way that each cell has a single digit, and each digit appears only once in every row, column, and $3\times 3$ subgrid.

Sudokus that appear in books, newspapers, etc often have the following characteristics.
\begin{enumerate}
\item The arrangement of initial digits (called \emph{clues}) forms a regular pattern.
\item The level of difficulty is moderate.
\end{enumerate}

\begin{figure}[H]
  \centering
  \includegraphics[keepaspectratio, width=6cm]{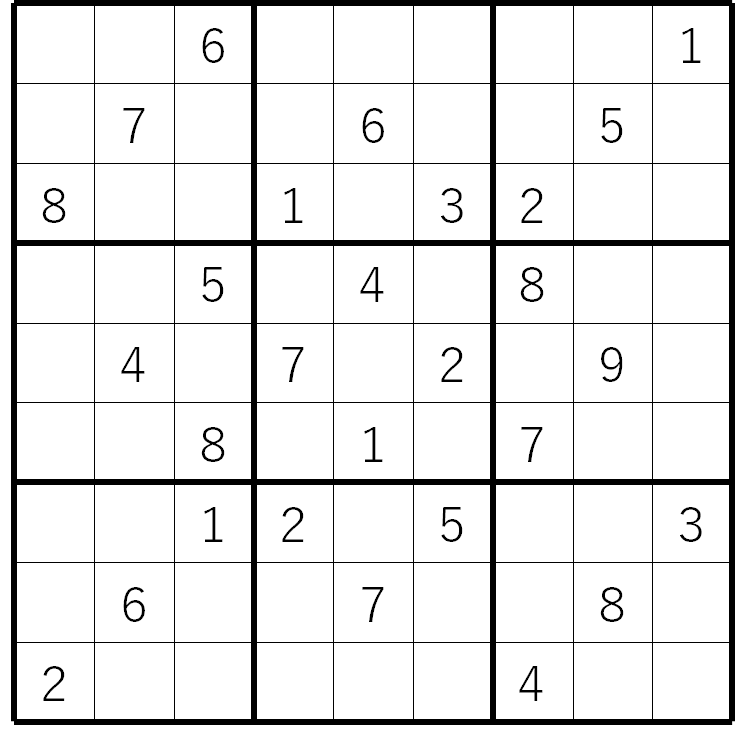}
  \includegraphics[keepaspectratio, width=6cm]{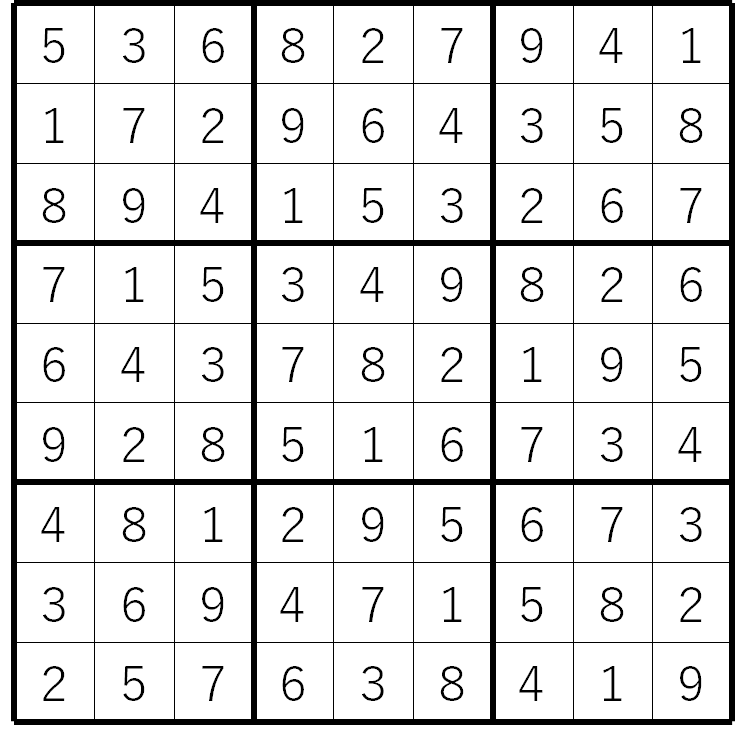}
  \caption{Sudoku clues (left) and its solution (right)}
  \label{fig:sudoku}
\end{figure}

Regarding~1, for example, the left grid of Figure~\ref{fig:sudoku} has reflection symmetry on the two diagonal axis.
It may also has the shape of a number, a letter, or a symbol.
Regarding~2, completing Sudokus generally requires backtracking, which is amenable to computers.
On the other hand, there is a set of techniques (called \emph{strategies}) that humans use in solving Sudokus by hand~\cite{mathsudoku}.
Typically strategies are if-then rules:
If a strategy is applicable to the current grid, then it might rule out digits as candidates or might determine digits as those to be finally placed.
A \emph{strategy-solvable Sudoku} is a partially completed grid that can be completed by applying strategies repeatedly.
Here, the strategies must be selected from a set of predetermined ones.
Since there is no need to guess, strategy-solvable Sudokus (at least for a few basic strategies) might be referred to as Sudokus having the moderate level of difficulty for humans.

Making new Sudoku instances having these characteristics is not an easy task.
There is a public software~\footnote{The generator of Zama and Sasano~\cite{weko_157970_1}~\url{http://www.cs.ise.shibaura-it.ac.jp/2016-GI-35/SudokuGenerator.tar.gz} , accessed May 26th, 2020.}, which is able to generate Sudokus from a given set of clue positions so that they are solvable with some known strategies.
However, it is still hard to generate those with $20$ or less clues and those with around $45$ or more clues.
For such instances, there might be no other choice but to rely on human intelligence involving the intuition, the inspiration and the experience of enthusiasts.
%Of course, it has a limit, and all the more so with $16\times 16$ and $25\times 25$ grids.

In this paper, we consider a method for determining Sudoku clues in specified positions such that all empty cells can be filled in with a specified set of strategies (see Figure~\ref{fig:scg}).

\begin{figure}[H]
  \centering
  \includegraphics[keepaspectratio, width=6cm]{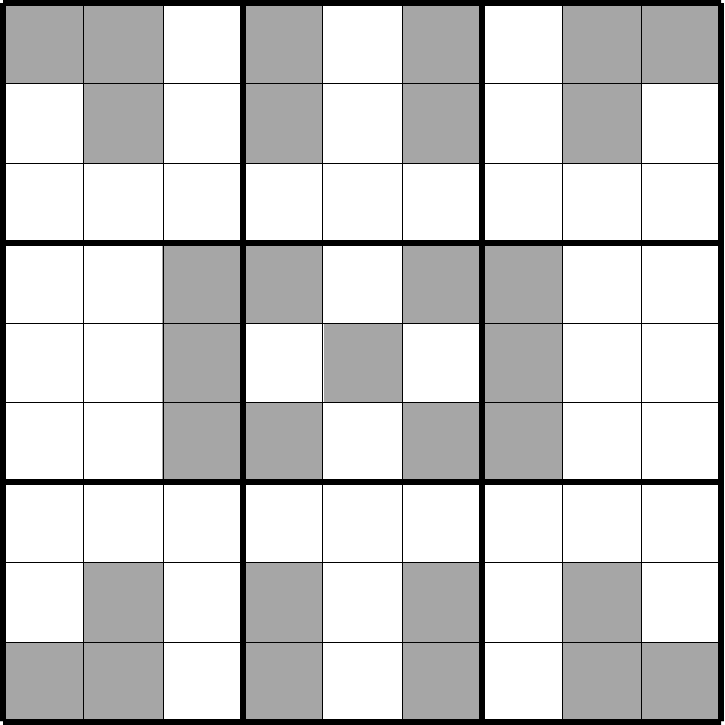}
  \includegraphics[keepaspectratio, width=6cm]{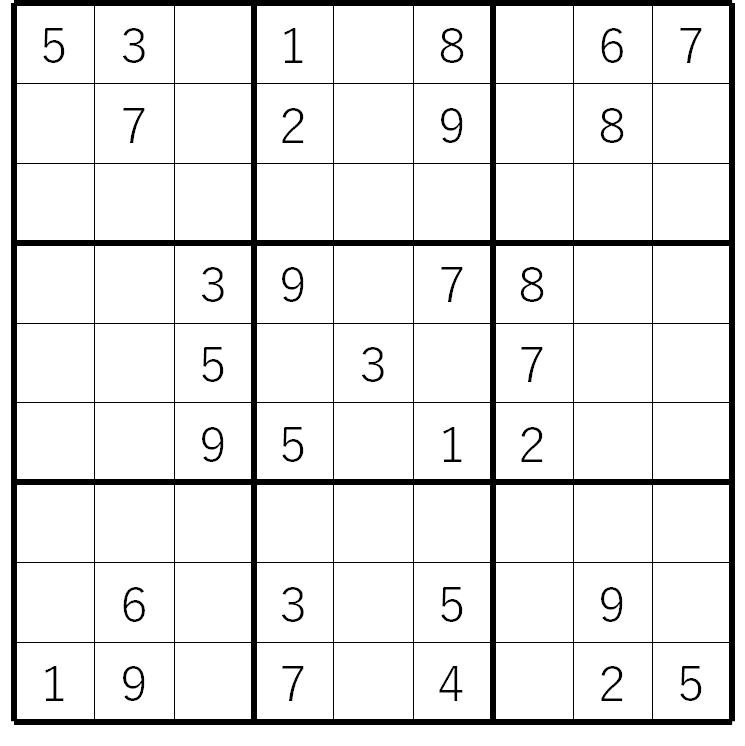}
  \caption{Clue positions (left) and clues solvable with naked singles only (right)}
  \label{fig:scg}
\end{figure}

Almost all Sudoku generators simply output proper Sudoku instances (i.e.\ partially completed grids having unique solutions).
Clue arrangements and strategies necessary for solving are uncontrollable.
To the best of our knowledge, the only exception is the generator of Zama and Sasano, mentioned earlier.
Their generator is based on the generate-and-test method, which repeats the followings until the test is passed or the number of trials exceeds a predetermined limit.
\begin{enumerate}
\item Generate clues in specified positions.
\item Test whether all the other cells are completed using specified strategies only.
\end{enumerate}
The same idea is also mentioned in~\cite{weko_138032_1}.
Since the test step can be done quickly, the key is to devise a criterion for the generation step so that the number of trials is as small as possible.

Since the generate-and-test method examines only a limited portion of the vast search space, it has the drawbacks listed below.
\begin{itemize}
\item The lower the density of solutions (i.e.\ strategy-solvable clues) over the whole search space becomes, the harder it becomes to find a solution.
\item If a set of clue positions happens to be not strategy-solvable, it is unable to recognize it no matter how much time passes.
\item Even for a strategy-solvable set of clue positions, there is no guarantee for being able to find a solution in a finite amount of time.
\end{itemize}

To tackle these issues (in particular the last two), we consider it necessary to formulate the concept of strategy-solvability.
It seems that strategy-solvability has been recognized intuitively, and no formal treatment has been given so far.
In this paper, we introduce a rigorous framework to discuss solvability for Sudoku instances with respect to strategies and define the notion of strategy-solvable Sudoku clues.
This allows us to handle not only known strategies but also general strategies under a few reasonable assumptions.
We then propose an exact method for determining strategy-solvable Sudoku clues for a given set of clue positions that is solvable with a given set of strategies.
The key is a reduction to a constraint satisfaction problem (CSP).
Our method is able to benefit from the power of state-of-the-art CSP solvers.
Pruning techniques of CSP solvers are expected to be effective for the first issue above.
Moreover, as long as any complete CSP solver is utilized, it is guaranteed that if a set of clue positions is not strategy-solvable, our method eventually recognizes the unsolvability; otherwise, our method eventually finds strategy-solvable clues.

This is the first exact method except for a trivial brute-force search.
As strategies allowed in completing Sudokus, this paper employs naked singles, hidden singles, and locked candidates~\cite{mathsudoku}, which are quite basic yet enough powerful to complete many Sudokus.
Indeed, we have confirmed that the combination of the three strategies allows for solving as many as $37,373$ of the $49,151$ minimum Sudokus (i.e.,\ Sudokus with $17$ clues) collected by Gordon F.\ Royle~\cite{minimumsudoku}.
Here we remark that our CSP formulation is almost independent of specific strategies.
In order to allow other strategies, it is sufficient to formulate the corresponding logical constraints and add them with a minor modification of the strategy-independent part.

Besides the clue generation, we present an application of our method to the problem of determining the minimum number of Sudoku clues that are solvable with a given set of strategies.
We demonstrate that our method is easily customized for this problem with a small modification in the CSP constraints.
It is well-known that there is no proper Sudoku with $16$ or less clues~\cite{There16}.
There are many Sudokus with $17$ clues that are solvable with basic strategies such as hidden singles.
Interestingly, this may not be the case for naked singles, arguably one of the most basic strategy, as we have confirmed that no grid solvable with only naked singles is included in the minimum Sudoku collection.
This poses an open problem of whether there is a gap between the minimum numbers for strategy-solvable Sudokus and proper Sudokus.
Our method will be useful in tackling this problem thanks to the ability of determining unsolvability.

We conduct experiments to compare our method with the generator of Zama and Sasano~\cite{weko_157970_1}, using grids varying the position and the number of clues at random.
From the results, we observe that our method terminates within $1$ minutes for many instances, showing our method being stable in terms of running time.
However, as the number of clues gets closer to $20$, the running time rapidly increases and exceeds a time limit.
On the other hand, the generator of Zama and Sasano often can find solutions much faster even in near $20$ clues, while the performance sharply deteriorates around $45$ clues and exceeds the time limit for all grids with more clue positions.
Perhaps this is due to the exponential blow-up of the search space (in other words, the exponential decline in the density of solutions).
We also evaluate our method for several instances with $17$ clue positions taken from known minimum Sudokus to see the efficiency for deciding unsolvability.

The main contributions of the present paper are to formulate the concept of strategy-solvability, to establish an exact method for the strategy-solvable Sudoku clues problem, and to demonstrate the flexibility of the CSP-based approach.
It remains as future work to improve our method in less clues.

Unless otherwise noted, the size of Sudoku is fixed to $9\times 9$ throughout the paper.
This is simply for convenience and our method can be easily translated into general Sudokus on $n^2\times n^2$ grid.

The paper is organized as follows.
Section~\ref{sec:preliminaries} introduces necessary notations, terminology, and explains strategies.
Section~\ref{sec:formulation} formulates the concept of strategy-solvability and the strategy-solvable Sudoku clues problem.
Section~\ref{sec:ourmethod} proposes an exact method for the strategy-solvable Sudoku clues problem, and Section~\ref{sec:improvements} presents two improvements for our method.
Section~\ref{sec:application} presents an application to the strategy-solvable minimum Sudoku problem.
Section~\ref{sec:experiments} presents experimental results.
Section~\ref{sec:conclusion} concludes this paper.

\section{Preliminaries}\label{sec:preliminaries}
In this section we introduce necessary notations and terminology, and we explain strategies.

\subsection{Notations and Terminology}
For convenience, rows and columns are numbered from $0$ to $8$.
A cell is denoted by the pair $\cell{i}{j}$ of a row index $i$ and a column index $j$.
The $9$ subgrids of $3\times 3$ are called \emph{blocks}.
Rows, columns, and blocks are collectively called \emph{groups}.
A group is identified with the set of all cells in the group.
By abuse of notation, we denote by $G\setminus \cell{i}{j}$ the difference of a singleton $\{\cell{i}{j}\}$ from a group $G$.

Given a partially completed grid, the goal of a Sudoku puzzle is to fill in all empty cells with digits from $1$ to $9$ in such a way that
\begin{itemize}
\item each cell has a single digit, and
\item each digit appears only once in every group.
\end{itemize}
The completed grid is called a \emph{solution}.
A \emph{Sudoku} is a convenient alias for a partially completed grid given as an initial grid.
A Sudoku is \emph{proper} if it has a unique solution.
The occurrence of a digit in an initial grid is called a \emph{clue}, and \emph{Sudoku clues} are the clues in an initial grid.
A cell in an initial grid to which some digit is designated to be placed as a clue is called a \emph{clue cell} or a \emph{clue position}.

\subsection{Strategies}
There is a set of techniques (called \emph{strategies}) that humans use in solving Sudokus by hand~\cite{mathsudoku}.
Typically strategies are if-then rules:
If a strategy is applicable to the current grid, then it might rule out digits as candidates or might determine digits as those to be finally placed.
Throughout the following explanation, let $n$ be a nonzero digit and $\cell{i}{j}$ be a cell.

A \emph{naked single} is a strategy that places $n$ in $\cell{i}{j}$ if no other candidate but $n$ remains at $\cell{i}{j}$.
For example, let us look at the gray cell $\cell{4}{7}$ in the left grid of Figure~\ref{fig:unsolvable}.
Since all digits but $5$ appear in either group having $\cell{4}{7}$, these digits must be ruled out and only $5$ remains.
Hence, $5$ is placed in $\cell{4}{7}$.
Starting with the left grid of Figure~\ref{fig:unsolvable}, one can reach the right grid using only naked singles, but cannot proceed any more because of two or more candidates over all empty cells.

\begin{figure}
  \centering
  \includegraphics[keepaspectratio, width=6cm]{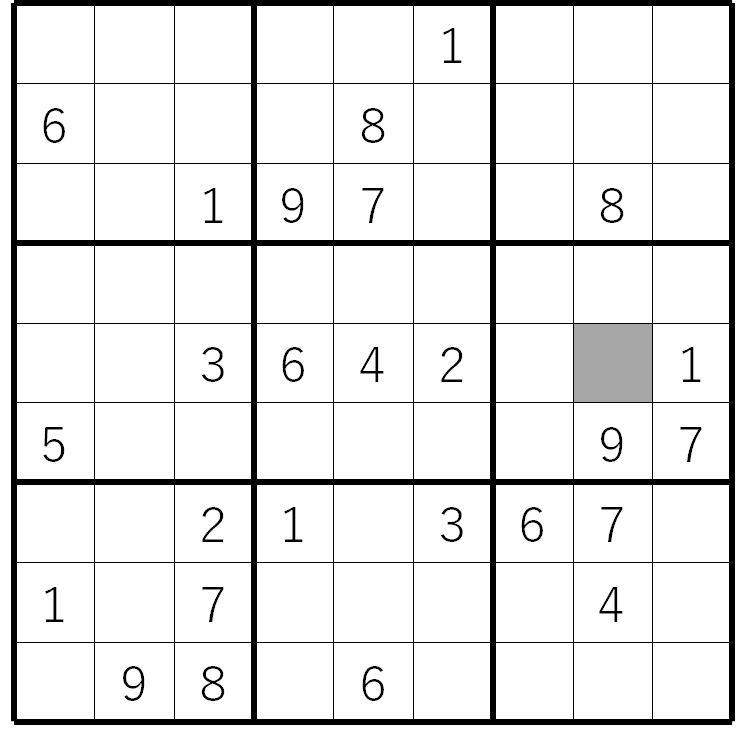}
  \includegraphics[keepaspectratio, width=6cm]{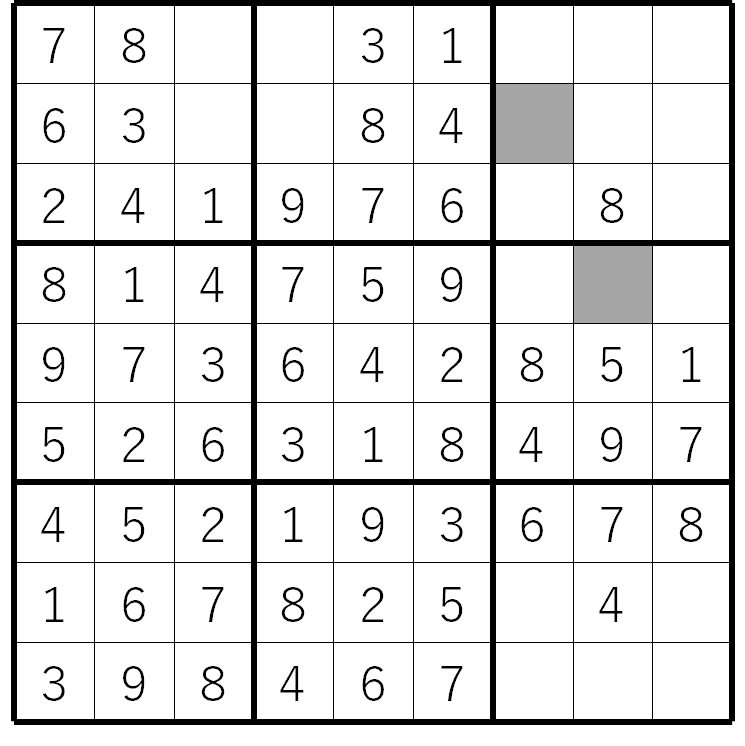}
  \caption{Starting with the partially completed grid on the left, one can reach the grid on the right using only naked singles but cannot proceed any more.}
  \label{fig:unsolvable}
\end{figure}

A \emph{hidden single} is a strategy that places $n$ in $\cell{i}{j}$ if there is a group $G$ having $\cell{i}{j}$ such that no cell in $G\setminus\cell{i}{j}$ has $n$ as a candidate.
For example, let us look at the right grid of Figure~\ref{fig:unsolvable}.
For the row of index $1$, we can observe that every cell in the row except for the gray cell $\cell{1}{6}$ does not have $7$ as a candidate.
Indeed, for any such cell, $7$ already appears in another cell of the same column.
Hence, a hidden single strategy determines $7$ in $\cell{1}{6}$.
In this way, by applying hidden singles repeatedly, the grid can be completed.

Let $A,B$ be groups such that $|A\cap B|=3$.
A \emph{locked candidate} is a strategy that rules out $n$ over all cells in one difference set $B\setminus A$ if no cell in the other difference set $A\setminus B$ has $n$ as a candidate.
For example, let us look at the right grid of Figure~\ref{fig:unsolvable}.
Let $A$ be the column of index $7$, and let $B$ be the block adjacent,  on the right, to the center block.
No empty cell in $A\setminus B$ has $3$ as a candidate because for each such cell, there is another cell in the same row to which $3$ is already placed.
Hence, $3$ is ruled out for all empty cells in $B\setminus A$.
Because of this, $2$ becomes a unique candidate at $\cell{3}{6}$.
It is determined as the digit at $\cell{3}{6}$ by a naked single.
In this way, the grid also can be completed using locked candidates as well as naked singles.

\section{Formulation}\label{sec:formulation}
Although some known strategies are introduced in the previous section and some solvable cases are explained using particular grids, there are still unclear points such as what is a strategy in general and what operations are allowed for placing or ruling out digits, and so on.
In this section, we thus introduce a rigorous framework to discuss solvability for Sudoku instances.
This allows us to handle not only known strategies but also general strategies under a few reasonable assumptions.
This also makes it clear that, for example, it is not allowed to temporarily place digits and then cancel them when a contradiction occurs.
Such an inference would allow us to do backtracking, which nullifies the restriction of completion methods.

\begin{figure}[H]
\centering
\includegraphics[keepaspectratio, width=2.4cm]{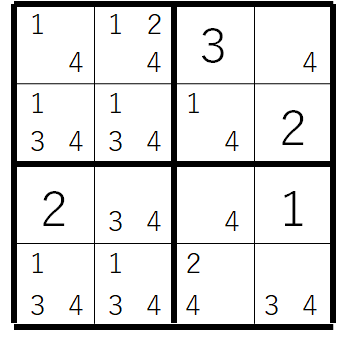}
\includegraphics[keepaspectratio, width=2.4cm]{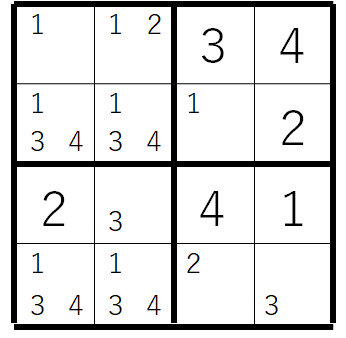}
\includegraphics[keepaspectratio, width=2.4cm]{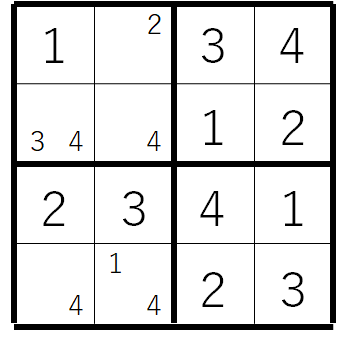}
\includegraphics[keepaspectratio, width=2.4cm]{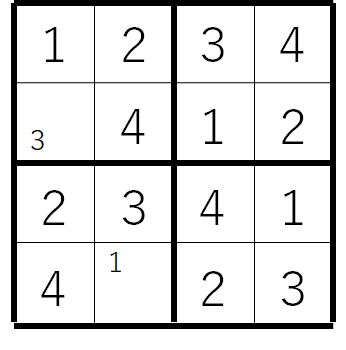}
\includegraphics[keepaspectratio, width=2.4cm]{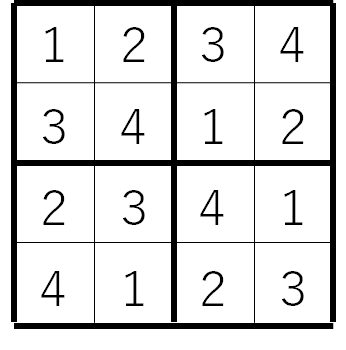}
\caption{Sequence of grids in $4\times 4$ Sudoku, evolving from the left-most initial state to the right-most final state.Larger digits are determined ones and smaller digits are candidates.}
\label{fig:transition}
\end{figure}

A \emph{state transition model} of Sudoku is a tuple $(Q, R, I, F)$ such that $Q$ is the set of all possible states, $R$ is a state transition relation, $I$ is the set of initial states, and  $F$ is the set of final states.
We will explain each component below.
Figure~\ref{fig:transition} shows how a grid evolves as strategies are applied, which will be used as an example in the succeeding explanation.

A \emph{state of grid} (or a \emph{state} in short) is a pair $\cell{f}{g}$ of functions.
The function $f$ maps the set of cells to the set of digits, and $\assign{f}{i}{j}=n$ means that if $n\not= 0$, then $n$ is placed in $\cell{i}{j}$; otherwise, no digit is placed.
The function $g$ maps the set of cells to the powerset of the set of nonzero digits, and $\assign{g}{i}{j}$ represents the set of all candidates at $\cell{i}{j}$.
Here, by abuse of notation, $\assign{f}{i}{j}$ and $\assign{g}{i}{j}$ denote $f\left (\cell{i}{j} \right )$ and $g\left (\cell{i}{j} \right )$, respectively.
All states must satisfy the conditions $S_1$ and $S_2$ below: for all cells $\cell{i}{j}$ and nonzero digits $n$, 
\begin{enumerate}
\item[$S_1$] $\assign{g}{i}{j}\not = \emptyset$,
%\item[$S_1$] $\assign{f}{i}{j}=n$ always implies $n\in\assign{g}{i}{j}$, 
\item[$S_2$] $n\not\in\assign{g}{i}{j}$ if
there is a different nonzero digit $n'$ from $n$ such that $\assign{f}{i}{j}=n'$ or
there is a different cell $\cell{i'}{j'}$ from $\cell{i}{j}$ with both in a common group such that $\assign{f}{i'}{j'}=n$.
\end{enumerate}
We denote by $Q$ the set of all states.

%The condition $S_1$ is simply for excluding the contradictory situation for which $n$ is placed in $\cell{i}{j}$ by one strategy while $n$ is ruled out by another strategy.
%We can safely assume $S_1$ because once violated, there is no chance of reaching any solution from the current grid.
The antecedent of $S_2$ is a condition for digits $n$ that can be safely ruled out whenever some digits are placed.
Imposing $S_2$ ensures that all such digits are ruled out.
This is simply to make all states having a type of normal form, and it would not deduce any wrong consequence as long as decisions for number-placements are correct such as not violating the properness.
In Figure~\ref{fig:transition} we can observe that all grids are in "normal form".

A \emph{state transition relation} is a binary relation $R\subseteq Q\times Q$.
The $\cell{q}{q'}\in R$ (or $qRq'$ in infix notation) means that there are strategies by which $q$ is changed to $q'$.
We refer to $q$ and $q'$ as the \emph{source state} and the \emph{target state} of the transition, respectively.
All state transitions must satisfy that for any cell $\cell{i}{j}$ and nonzero digit $n$, if $\assign{f}{i}{j}=n$ holds in the source state, so does in the target state; if $n\not\in\assign{g}{i}{j}$ holds in the source state, so does in the target state.
This simply means that once digits are placed or ruled out, then it cannot be canceled afterward.
We put one more assumption as follows.
For any three states $q_1=\cell{f_1}{g_1}$, $q_2=\cell{f_2}{g_2}$, and $q_3=\cell{f_3}{g_3}$ with $q_1 R q_2$ and $q_2 R q_3$, if $g_1=g_2$, then $q_2=q_3$.
We consider this assumption not so strong.
Indeed, many strategies such as naked singles and hidden singles decide to place digits by examining only candidates.
For such strategies, suppose that $q_2$ is obtained from $q_1$ by applying all strategies applicable to $q_1$.
Then, if $g_1=g_2$, no strategy applicable to $q_2$ remains, and $q_2$ cannot be changed any more, i.e., $q_2=q_3$.

A state $q=\cell{f}{g}\in Q$ is an \emph{initial state} if it satisfies the followings: $\assign{f}{i}{j}\not = 0$ for all clue cells; $\assign{f}{i}{j}=0$ for the other cells; $n\not\in\assign{g}{i}{j}$ if and only if the antecedent of $S_2$ holds.
The last condition ensures that for each clue cell, all digits but the determined one are ruled out; all digits that appear as candidates in a common group to some determined digit are ruled out; (the only-if part) no other digits are ruled out.
Notice that the only-if part is the most essential because the if part can be omitted thanks to $S_2$.
By the way, for non-initial states, this must not be assumed in general because of strategies for number-removals such as locked candidates.
For example, in the left most grid of Figure~\ref{fig:transition} we can observe that all digits not designated to be ruled out by the antecedent of $S_2$ remain as candidates.
We denote by $I$ the set of all initial states.

A \emph{final state} is a state such that digits are determined for all cells, that is, $\assign{f}{i}{j}\not =0$ for all $\cell{i}{j}$.
We denote by $F$ the set of all final states.

\begin{defi}\label{def:solvability}
Let $(Q,R,I,F)$ be a state transition model of Sudoku.
A state $q_0$ is \emph{strategy-solvable} if there is a sequence of states $q_0,\ldots,q_k\in Q$ such that $q_0\in I$, $q_k\in F$, and $q_i R q_{i+1}$ for all $i\in\{0,\ldots,k-1\}$.
\end{defi}

\begin{defi}
Let $(Q,R,I,F)$ be a state transition model of Sudoku.
The \emph{Strategy-solvable Sudoku Clues} problem (or \emph{SSC} in short) is to determine whether there is a strategy-solvable state of $(Q,R,I,F)$.
\end{defi}

Suppose that $R$ is given as a logical formula that represents strategies and a set $S$ of clue cells is also given.
Consider $Q$, $I$, and $F$ as the ones defined as described above.
Sudoku clues are \emph{strategy-solvable} if the initial state corresponding to the Sudoku clues is strategy-solvable with respect to $(Q,R,I,F)$.
The strategy-solvable Sudoku clues problem is defined accordingly.
If the output of the problem is yes, in practice we will also ask for generating strategy-solvable Sudoku clues.

%The following proposition states that digits to be determined must be selected from candidates.
%
%\begin{prop}\label{prop:s0}
%Let $q=\cell{f}{g}$ be a state of grid.
%For any cell $\cell{i}{j}$ and nonzero digit $n$, if $\assign{f}{i}{j} = n$, then $n\in\assign{g}{i}{j}$.
%\end{prop}
%
%\begin{proof}
%Suppose that $\assign{f}{i}{j}=n$ and $n\not\in\assign{g}{i}{j}$.
%Applying $S_2$, we obtain $n'\not\in\assign{g}{i}{j}$ for all $n'\not = n$.
%This means $\assign{g}{i}{j}=\emptyset$, which is contradictory to $S_1$.
%\end{proof}

\begin{prop}\label{prop:state}
Let $q=\cell{f}{g}$ be a state of grid.
For any cells $\cell{i}{j}$, $\cell{i'}{j'}$ in a common group, $\assign{f}{i}{j}=\assign{f}{i'}{j'}$ implies $\cell{i}{j}=\cell{i'}{j'}$ or $\assign{f}{i}{j}=0$.
\end{prop}

\begin{proof}
Suppose that $\assign{f}{i}{j}=\assign{f}{i'}{j'}$, $\cell{i}{j}\not = \cell{i'}{j'}$, and $\assign{f}{i}{j}\not = 0$.
Let $n=\assign{f}{i'}{j'}$.
Applying $S_2$, we obtain $n\not\in\assign{g}{i}{j}$.
Since $n=\assign{f}{i}{j}$, applying $S_2$, we obtain $n'\not\in\assign{g}{i}{j}$ for all $n'\not= n$.
From both, $\assign{g}{i}{j}=\emptyset$ follows, and this is contradictory to $S_1$.
\end{proof}

\begin{cor}
Every final state represents a Sudoku solution.
\end{cor}

\begin{proof}
Let $q=\cell{f}{g}$ be a final state.
Since $f$ is a function, each cell must have a single digit.
From Proposition~\ref{prop:state} and $\assign{f}{i}{j}>0$, it follows that each digit appears once in every group.
\end{proof}

Notice that this corollary holds no matter what state transition relation is given.

\begin{prop}\label{prop:bound}
Let $R$ be a state transition relation.
Let $q_{0}=\cell{f_0}{g_0},\ldots,q_{k}=\cell{f_k}{g_k}$ be a sequence of states such that $q_s R q_{s+1}$ for all $s\in\{0,\ldots,k-1\}$.
If $g_s\not =g_{s+1}$ for all $s\in\{0,\ldots,k-1\}$, then $k\leq 648$.
\end{prop}

\begin{proof}
Let $m_s$ be the number of all triples $\left (i,j,n\right )$ such that $n\in\assign{g_s}{i}{j}$.
Clearly, $729=9^3\geq m_0>m_1>\cdots>m_k\geq 81$.
Hence $k\leq 648$.
\end{proof}

\section{Our Method}\label{sec:ourmethod}
In this section, we propose an exact method for the strategy-solvable Sudoku clues problem (SSC).
As strategies allowed in completing Sudokus, this paper employs naked singles, hidden singles, and locked candidates, which are quite basic yet enough powerful to complete many Sudokus.

\subsection{Overview}
The key of our method is to encode a given SSC instance to an equivalent constraint satisfaction problem (CSP) instance and then to solve it using a generic CSP solver.
For the CSP encoding, we introduce variables for representing states of grid and other auxiliary variables.
Using such variables, we represent constraints for states, state transitions, particular strategies, and so on.

Any variable assignment that satisfies all such constraints substantially corresponds to a sequence of states $q_0,\ldots,q_k$ with $q_0$ strategy-solvable, as characterized in Definition~\ref{def:solvability}.
Here, the initial state $q_0$ represents an assignment of digits to clue cells and the whole sequence represents a history of evolving grids and applied strategies, which eventually reaches a Sudoku solution.
Hence, by applying a CSP solver to the encoded constraints, we eventually obtain strategy-solvable Sudoku clues if exist; otherwise, we eventually recognize the strategy-unsolvability of the clue cells, that is, no assignment of digits to the clue cells is strategy-solvable.

The details of our method are explained as follows.
At first, we present constraints for a state transition framework, which is almost independent of particular strategies.
As stated above, we introduce variables for representing states of grid and auxiliary variables.
We then represent various constraints using such variables.
After that, we shift to particular strategies. we present constraints for number-placements and constraints for number-removals.
We finally provide the whole picture of our method as well as some remarks.

\subsection{Constraints for A General State Transition Framework}
We introduce variables for representing states of grid.
We then present constraints for initial states, constraints for state transitions, constraints for final states.

\subsubsection{Variables}
Let $q_k=(f_k,g_k)$ be a state of grid in step $k\geq 0$, i.e., a state reachable from an initial state by $k$ transitions.
Let $i$ and $j$ be a row index and a column index, respectively.
The integer variable $\xvar{i}{j}{k}$ encodes the value of $f_k$ at cell $\cell{i}{j}$.
In other words, $\xvar{i}{j}{k}$ takes $0$ if no digit is placed in $\cell{i}{j}$, and it takes a nonzero digit $n$ if $n$ is placed in $\cell{i}{j}$.
The Boolean variable $\yvar{i}{j}{n}{k}$ encodes whether the set of candidates, $\assign{g_k}{i}{j}$, at cell $\cell{i}{j}$ includes a nonzero digit $n$.
Here $\yvar{i}{j}{n}{k}$ is true if and only if $n\in\assign{g_k}{i}{j}$.
We call variables of the forms $\xvar{i}{j}{k}$ and $\yvar{i}{j}{n}{k}$ \emph{$X$-variables} and \emph{$Y$-variables}, respectively.

The remaining variables are Boolean variables of the form $\zvar{m}$, which are used for auxiliary purpose.
We call them \emph{$Z$-variables}.
Each $Z$-variable is introduced per particular condition for a number-placement or a number-removal.
The parameter $m$ is simply an identifier for distinguishing between a large number of particular conditions.
We thus number all such conditions in an arbitrary order.

By the way, since $X$-variables and $Y$-variables have $k$ in their parameters, we have to fix the maximum number of steps, $K$, in advance.
In principle, it suffices to let $K=649$.
This is shown later.
Since this is too large in practice, we will propose a practical method for doing with a smaller $K$ while keeping the exactness in Section~\ref{sec:improvements}.

\subsubsection{Constraints for Initial States}
The constraints for initial states are as follows.
For all clue cells $\cell{i}{j}$,
\begin{equation}\label{eq:init1}
\xvar{i}{j}{0} \not = 0,
\end{equation}
and for the other cells $\cell{i}{j}$,
\begin{equation}\label{eq:init2}
\xvar{i}{j}{0} = 0.
\end{equation}

Here we would not explicitly impose any constraint on $Y$-variables in step $0$ because such constrains can be derived from the constraints for $Y$-variables of an arbitrary step and the constraints for number-removals.

\subsubsection{Constraints for State Transitions}
The constraints for state transitions represent how $X$-variables and $Y$-variables in the current grid are determined from the previous grid.

The constraint that allows us to determine a nonzero digit $n$ at cell $\cell{i}{j}$ in step $k$ is as follows.

\begin{equation}\label{eq:xvar}
\xvar{i}{j}{k} = n \leftrightarrow \bigvee \zvar{m}
\end{equation}
Here, $\zvar{m}$ runs over all possible $Z$-variables that can derive the left hand side.
The only-if part means that if $\xvar{i}{j}{k}=n$, then there must be an evidence, i.e., some $Z$-variable is true.
One such $Z$-variable is the case for which $n$ is determined at $\cell{i}{j}$ in the previous step.

\begin{equation}\label{eq:xvar_prev}
\zvar{m} \leftrightarrow \xvar{i}{j}{k-1} = n
\end{equation}
As noticed before, suppose that the parameter $m$ is uniquely determined in order to distinguish the condition in the right hand side from the other particular conditions for $Z$-variables.
Hereafter we will follow the same implicit agreement whenever $\zvar{m}$ appears in constraints.
The other $Z$-variables of Formula~\ref{eq:xvar} will be detailed later according to particular cases.
Notice that as discussed in Section~\ref{sec:formulation} in order for our assumption for state transitions to be met, it is sufficient to let $\zvar{m}$ be determined by only $Y$-variables in the previous step (except for Formula~\ref{eq:xvar_prev}).

The constraint that allows us to rule out a nonzero digit $n$ as a candidate at cell $\cell{i}{j}$ in step $k$ is as follows.

\begin{equation}\label{eq:yvar}
\lnot\yvar{i}{j}{n}{k} \leftrightarrow \bigvee\zvar{m}
\end{equation}

In the same way as $X$-variables, $\zvar{m}$ runs over all possible $Z$-variables that can derive the left hand side.
One such $Z$-variable is the case for which $n$ is ruled out as a candidate at $\cell{i}{j}$ in the previous step.

\begin{equation}
\zvar{m} \leftrightarrow \lnot\yvar{i}{j}{n}{k-1}
\end{equation}
The other $Z$-variables of Formula~\ref{eq:yvar} will be detailed later according to particular cases.

\subsubsection{Constraints for Final States}
The constraints for final states are those that allow us to reject the current state as soon as it turns out that there is no chance to reach final states.
There are three cases for the current state of grid.
\begin{enumerate}
\item All cells are completed.
\item Some empty cells remain, and some candidates have been ruled out from the previous grid.
\item Some empty cells remain, and no candidate has been ruled out from the previous grid.
\end{enumerate}
Clearly the first case must be accepted.
For the third case, our assumption of state transitions (see Section~\ref{sec:formulation}) implies that the current incomplete grid cannot be changed anymore.
Hence, the third case must be rejected immediately.
For the second case, since applicable strategies may remain for the current grid, the decision must be postponed to the next grid.
The following formula is in change of this filtration.

\begin{equation}\label{eq:final}
\bigwedge_{\substack{0\leq i,j<9\\ 1\leq n\leq 9}}\Bigl ( \yvar{i}{j}{n}{k-1}\leftrightarrow \yvar{i}{j}{n}{k} \Bigr )
\rightarrow \bigwedge_{0\leq i,j <9} \xvar{i}{j}{k}\not=0 
\end{equation}

Suppose that the current grid is in step $k\ (\geq 1)$.
The left hand side means that no candidate has not been ruled out from the previous grid, and the right hand side means that the current grid is completed.
It is clear that this formula rejects only the third case among the three cases.

As announced earlier, we here show that $K=649$ is sufficient.
Let $q_0=\cell{f_0}{g_0},\ldots,q_k=\cell{f_k}{g_k},\ldots$ be any sequence of states such that $q_0$ is an initial state and $q_k$ is the next state of $q_{k-1}$ for all $k$.
The assumption for state transitions (see Section~\ref{sec:formulation}) ensures that once $g_{\bar{k}-1}=g_{\bar{k}}$ for some $\bar{k}$, it follows that $q_{\bar{k}}=q_k$ for all $k\geq \bar{k}$.
From Proposition~\ref{prop:bound}, the minimum $\bar{k}$ with $g_{\bar{k}-1}=g_{\bar{k}}$ is bounded from above by $649$.
Therefore it suffices to let $K=649$.
Notice that $\bar{k}\leq K$ is necessary in order to decide whether $q_{\bar{k}}$ is a final state.

\subsection{Constraints for Particular Number-placements}
We explain particular cases for $\zvar{m}$s in Formula~\ref{eq:xvar} that we have put off.

\subsubsection{Naked Singles}
The right hand side of the following formula is the condition for which a naked single allows us to determine a nonzero digit $n$ at cell $\cell{i}{j}$ in step $k\ (\geq 1)$.
Notice that in order for this to take effect, it is necessary that the following $\zvar{m}$ is registered to Formula~\ref{eq:xvar}.

\begin{equation}\label{eq:f}
\zvar{m} \leftrightarrow \bigwedge_{\substack{n\not = n'\\ 1\leq n'\leq 9}} \lnot\yvar{i}{j}{n'}{k-1}
\end{equation}

The right hand side means that all digits but $n$ are ruled out as candidates at cell $\cell{i}{j}$ in the previous step.
There are two cases: either only $n$ remains as a candidate or no digit remains.
Although a naked single is applicable only to the former case, there is no substantial problem.
Because even if the latter case occurs, we obtain $\xvar{i}{j}{k}=n$ from Formulas~\ref{eq:f} and~\ref{eq:xvar}, and we also obtain $\xvar{i}{j}{k}=n'$ for any other digit $n'$ in the same way, which is a contradiction.

\subsubsection{Hidden Singles}
The right hand side of the following formula is the condition for which a hidden single allows us to determine a nonzero digit $n$ at cell $\cell{i}{j}$ in step $k\ (\geq 1)$.

\begin{equation}\label{eq:hs}
\zvar{m} \leftrightarrow \bigwedge_{\cell{i'}{j'}\in G\setminus\cell{i}{j}} \lnot\yvar{i'}{j'}{n}{k-1}
\end{equation}
Here, let $G$ be any group such that $\cell{i}{j}\in G$.
There are three cases for $G$, and for each case the corresponding constraint must be generated from Formula~\ref{eq:hs}.
The right hand side of Formula~\ref{eq:hs} does not exclude the case for which $n$ is ruled out over all cells in $G$, but there is no substantial problem, just like Formula~\ref{eq:f}.

\subsection{Constraints for Particular Number-removals}
We in turn explain particular cases for $\zvar{m}$s in Formula~\ref{eq:yvar} that we have put off.

\subsubsection{Conditions for States of Grid}
We have imposed the two conditions $S_1$ and $S_2$ on states of grid in order to reject states that violate the rules of a Sudoku puzzle.
We introduce constraints corresponding to $S_1$ and $S_2$.
Unlike the other constraints, the constraints below determine a relation between $X$-variables and $Y$-variables in the same step.

The following formula corresponds to $S_1$, which ensures that all cells $\cell{i}{j}$ have at least one candidate in all steps $k\geq 0$.

\begin{equation}\label{eq:s1}
\bigvee_{1\leq n\leq 9} \yvar{i}{j}{n}{k}
\end{equation}

The following formula in turn corresponds to one of the two cases in the antecedent of $S_2$, that is, a different digit $n'$ from $n$ is determined at cell $\cell{i}{j}$.

\begin{equation}\label{eq:s2-1}
\zvar{m} \leftrightarrow \bigvee_{\substack{n\not = n'\\ 0\leq i,j< 9}} \xvar{i}{j}{k}=n'
\end{equation}
Here, suppose that $\zvar{m}$ is registered in Formula~\ref{eq:yvar}.
If the right hand side of Formula~\ref{eq:s2-1} holds, then it follows from Formulas~\ref{eq:yvar} and~\ref{eq:s2-1} that $n$ is ruled out as a candidate at $\cell{i}{j}$ in the same step $k$, which is the consequent of $S_2$.

The following formula corresponds to the other case in the antecedent of $S_2$, that is, there is a different cell $\cell{i'}{j'}$ from $\cell{i}{j}$ with both in  a common group $G$ such that $n$ is determined at $\cell{i'}{j'}$.

\begin{equation}\label{eq:s2-2}
\zvar{m} \leftrightarrow \bigvee_{\cell{i'}{j'}\in G\setminus\cell{i}{j}} \xvar{i'}{j'}{k}=n
\end{equation}
Here, let $G$ be any group such that $\cell{i}{j}\in G$.
There are three cases for $G$, and for each case the corresponding constraint must be generated from Formula~\ref{eq:s2-2}.
Just like Formula~\ref{eq:s2-1}, Formula~\ref{eq:s2-2} together with Formula~\ref{eq:yvar} imply that $n$ is ruled out as a candidate at $\cell{i}{j}$ in the same step $k$, which is the consequent of $S_2$.

\subsubsection{Locked Candidates}
The right hand side of the following formula is the condition for which a locked candidate allows us to rule out a nonzero digit $n$ at cell $\cell{i}{j}$ in step $k\ (\geq 1)$.

\begin{equation}\label{eq:lc}
\zvar{m} \leftrightarrow \bigvee_{\cell{i'}{j'}\in A\setminus B} \lnot \yvar{i'}{j'}{n}{k-1}
\end{equation}
Here, let $A$ and $B$ be any groups such that $\cell{i}{j}\in B\setminus A$ and $|A\cap B|=3$.
For all possible combinations of $A$ and $B$, the corresponding constraint must be generated from Formula~\ref{eq:lc}.
Notice that Formula~\ref{eq:lc} can be commonly used to deduce $\yvar{i''}{j''}{n}{k}$ for all $\cell{i''}{j''}\in B\setminus A$, by which the size of constraints can be largely reduced.

\subsection{Remarks}
We have presented a CSP encoding for the strategy-solvable Sudoku clues problem.
That is, given a set of cells, there are clues for the cells such that the grid can be completed using naked singles, hidden singles, and locked candidates if and only if there is an assignment of $X$-variables, $Y$-variables, and $Z$-variables that satisfies all the constraints generated from Formulas~\ref{eq:init1} to~\ref{eq:lc}.
A satisfying assignment substantially corresponds to a sequence of states $q_0,\ldots,q_k$ with $q_0$ strategy-solvable, as characterized in Definition~\ref{def:solvability}.
By applying a CSP solver to the encoded constraints, we eventually obtain strategy-solvable Sudoku clues if exist; otherwise, we eventually recognize the strategy-unsolvability of the clue cells, that is, no assignment of digits to the clue cells is strategy-solvable.
The constraints for a general state transition framework are almost independent of particular strategies.
The only connection is via $Z$-variables in Formulas~\ref{eq:xvar} and~\ref{eq:yvar}.
In order to include additional strategies, it is sufficient to model these strategies as logical formulas independently and then register new $Z$-variables for them to Formulas~\ref{eq:xvar} or~\ref{eq:yvar}.

\section{Two Improvements}\label{sec:improvements}
In this section, we present two improvements for our method.

\subsection{Reduction of Constraint Size}
The biggest issue of our method is that there are hundreds of thousands of constraints.
A simple way for reducing a constraint size is to eliminate constraints concerning clue cells.
For all clue cells, digits are determined in an initial state, and the determined digits do not change over all succeeding steps.
Hence, all constraints necessary for clue cells in each step is take the same values as those in the previous step.
In the constraints for final states, there is no need to examine the values of $X$-variables and $Y$-variables for clue cells.
However, we must not ignore the constraints for clue cells in step $0$ corresponding to Formulas~\ref{eq:s1},~\ref{eq:s2-1}, and~\ref{eq:s2-2} because otherwise, we could not reject initial states violating the rules of Sudoku puzzle.

\subsection{Incremental Approach}
Our method requires to fix a maximum step size, $K$, in advance.
Since it has a significant impact on a constraint size, $K$ needs to be as small as possible.
However, for a small $K$, our method may return false clues.
To see this, let us take a look at Figure~\ref{fig:unsolvable}.
Suppose that the left grid is an initial state and the right grid is a state in step $K$.
The right grid is not yet completed, but no constraint has been violated up to step $K$.
Since there is no further step, our method returns the initial grid, even if only a naked single strategy is allowed.
Notice that if the right grid was in step $K-1$ or less, our method would reject the initial grid (because the constraints for final states are violated).

We propose a practical method for doing with a smaller $K$ while keeping the exactness.
A basic idea is to repeatedly apply our original method while incrementing $K$ from an initial number $K_{min}$ until $K$ exceeds a sufficiently large number $K_{max}$.
Algorithm~\ref{alg:inc} is a pseudo code for the method.
Notice that once constraints become unsatisfiable, constraints with any lager maximum step size must be unsatisfiable.

\begin{algorithm}[t]
  \caption{Incremental approach for the strategy-solvable Sudoku clues (SSC)}
  \begin{algorithmic}\label{alg:inc}
\FOR {$K=K_{min}$ to $K_{max}$}
\STATE Generate constraints for a given SSC instance with maximum step $K$.
\IF {the set of constraints is unsatisfiable}
\STATE return UNSAT
\ELSIF {the grid in step $K$ is completed}
\STATE return Sudoku clues in the initial grid.
\ENDIF
\ENDFOR
  \end{algorithmic}
\end{algorithm}

\section{Application}\label{sec:application}
Besides the clue generation, we present an application of our method to the problem of determining the minimum number of Sudoku clues that are solvable with a given set of strategies.

%It is well-known that there is no proper Sudoku with $16$ or less clues.
%There are many Sudokus with $17$ clues that are solvable with basic strategies such as hidden singles.
%Interestingly, this may not be the case for a naked single, arguably one of the most basic strategy, as we confirmed that no grid solvable with only naked singles is included in the minimum Sudoku collection of Gordon F.\ Royle.
%This poses an open problem of whether there is a gap between the minimum numbers for strategy-solvable Sudokus and proper Sudokus.
%Our method will be useful in tackling this problem thanks to the ability of determining unsolvability.

Our method is easily customized with a small modification as follows.
Remove Formulas~\ref{eq:init1} and~\ref{eq:init2}.
Instead, introduce integer variables taking $0$ or $1$, $\uvar{i}{j}$, for all cells $\cell{i}{j}$ and the following formula.
\begin{equation}
\uvar{i}{j} = 1 \leftrightarrow \xvar{i}{j}{0} \not = 0
\end{equation}
Introduce the following formula for ensuring that the number of determined digits in step $0$ is less than or equal to a threshold $\theta$.
\begin{equation}
\sum \uvar{i}{j} \leq \theta
\end{equation}
Here the summation runs over all variables $\uvar{i}{j}$.

In order to compute the minimum number, it is sufficient to repeatedly solve the modified constraints while decreasing $\theta$ one by one until the constraints become unsatisfiable.

\section{Experiments}\label{sec:experiments}
In this section, we conduct experiments to evaluate our method.

\subsection{Common Settings and Remarks}
Our method uses Sugar version 2.3.3~\cite{Sugar}~\cite{tamura2008sugar} as a CSP solver and MiniSat version 2.0~\cite{Minisat}~\cite{10.1007/978-3-540-24605-3_37} as a SAT solver internally invoked by Sugar.
The maximum step size $K$ is set to $30$, and our method is applied only once (i.e., no use of incremental approach) because in preliminary experiments, we have confirmed that $K=30$ is sufficient for many instances including those used in this experiments.
The computational environment is as follows.
\begin{itemize}
\item[] OS: Ubuntu 18.04.4 LTS
\item[] Main memory: 16GB
\item[] CPU: Intel~\textregistered\ Core~\texttrademark\ i7-4600U 2.10GHz
\end{itemize}

The implementation of our method is checked for all possible arrangements of $3$ and $4$ clue positions for $4\times 4$ Sudoku.
These instances are so small that a naive brute force search can quickly decide.
It is confirmed that the solvability results for our method completely coincide with those for the brute force search over all instances.
Here all strategies allowed to be used are naked singles, hidden singles, and locked candidates.
It turns out that there is no strategy-solvable instance with $3$ clue positions, and there are exactly $704$ strategy-solvable instances with $4$ clue positions.
By the way, all of the $704$ instances are also solvable with naked singles only.
Since there is no proper $4\times 4$ Sudoku with $3$ clues, there is no gap between the minimum numbers of proper Sudokus and strategy-solvable Sudokus with naked singles.

The implementation of (the CSP encoding part of) our method, tools such as the brute force search program, and all instances used in the experiments are publicly available in our website\footnote{\url{http://www.disc.lab.uec.ac.jp/toda/code/scg.html} , accessed in May 27th 2020.}.

\subsection{Running Time Comparison}
We compare our CSP-based method (CSP) with the generator of Zama and Sasano~\cite{weko_157970_1} (ZS).
The time limit is set to $600$ seconds.
In both methods, only naked singles, hidden singles, and locked candidates are allowed to be used.
Although other strategies are implemented in the generator of Zama and Sasano, we restricts the program to the three strategies by modifying their program.

We make total $100$ input instances (i.e.\ sets of cells) by repeating the following procedure:
select a number $n$ from $20$ to $79$ at random, and generate distinct $n$ cells at random.
Table~\ref{tab:data} shows the distribution of input instances with respect to the number of clues.
All instances are confirmed to be strategy-solvable.

\begin{table*}[t]
\begin{center}
\caption{Distribution of input instances with respect to the number of clues}
\begin{tabular}{l|cccccc}\hline\hline
\#clues& $20-29$ & $30-39$ & $40-49$ & $50-59$ & $60-69$ & $70-79$\\ \hline
\#instances  & $19$    & $30$    & $20$    & $10$    & $11$    & $10$
\end{tabular}
\label{tab:data}
\end{center}
\end{table*}

\begin{figure}[H]
  \centering
\includegraphics[keepaspectratio, width=12cm]{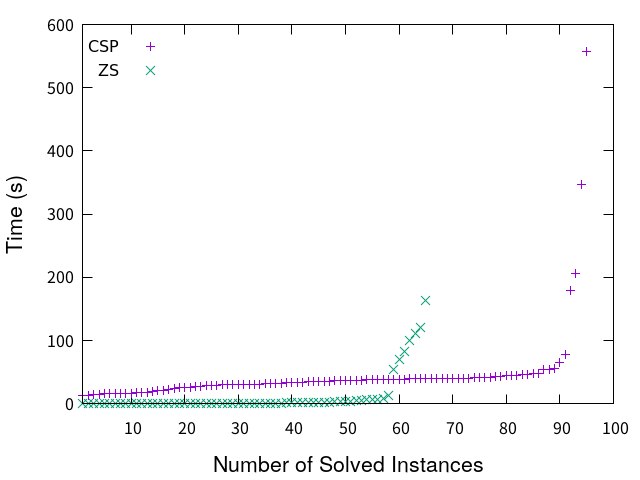}
\caption{Comparison of running time}
\label{graph}
\end{figure}

Figure~\ref{graph} shows a cactus plot of running time comparison.
Each point is plotted so that the $x$-coordinate is the number of instances solved within the time specified in the $y$-coordinate.
The points for our method (CSP) and the generator of Zama and Sasano (ZS) are indicated with $+$ and $\times$, respectively.
The curve formed by the points for the same method shows an increase in the number of solved instances over time.

Our method solves $95$ of total $100$ instances within the time limit, while the generator of Zama and Sasano solves $65$ instances.
The curve of our method shows a linear increase up to around $90$ in the $x$-axis.
All instances in this range are solved within about $1$ minutes.
The curve grows rapidly after $90$.
All instances in this rage as well as unsolved instances have $25$ or less cells.

On the other hand, the generator of Zama and Sasano often can find solutions much faster even in near $20$ cells as the curve of Zama-Sasano almost overlaps the $x$-axis up to near $60$.
All instances in this range have around $45$ or less cells.
After that, the performance drops significantly around $60$ as the curve sharply increases.
All instances with $50$ or more cells cannot be solved within the time limit.

\subsection{Evaluation on Unsolvability}
To see the efficiency for deciding unsolvablity, we make instances which seem to be on the border between solvable cases and unsolvable cases in such a way that we randomly select $30$ minimum Sudokus from the collection of Gordon F.\ Royle~\cite{minimumsudoku} and extract only clue positions, simply forgetting placed digits.
In this experiment, all strategies allowed to be used are naked singles only.
Notice that although no grid solvable with naked singles is included in the collection, it is possible that changing digits but in the same positions makes grids solvable.

Within several hours, our method terminates for $14$ of the $30$ instances, and all of them are confirmed to be not strategy-solvable.
Tables~\ref{tab:min30-1} and~\ref{tab:min30-2} in Appendix show the chosen minimum Sudokus and the running times.
Although $10$ instances take merely from $10$ to $20$ minutes, other instances take much more time.
It is still unknown whether there are $17$ clues solvable with naked singles only.
It appears far from the settlement of the problem on a gap betwen the numbers of strategy-solvable Sudokus and proper Sudokus by a computer program, considering a large number of possible arrangements of $17$ positions.

\section{Conclusion}\label{sec:conclusion}
Sudokus that appear in books, newspapers, etc often have a regular pattern in the arrangement of initial digits and they are typically made so that all empty cells can be completed using some known techniques, called strategies.
We formally defined the problem of generating such Sudoku instances by introducing the concept of strategy-solvability, which means that all empty cell can be filled in with digits using only a given set of strategies.
We proposed an exact method for solving this problem.
The key is to encode a given problem instance into an equivalent CSP instance and then solve it by applying a CSP solver.

There are a few existing researches, but all of them are based on the generate-and-test method, which repeats to generate a set of clues and then test whether it is strategy-solvable.
There are some drawbacks such as not being able to recognize that a specified set of cells is strategy-unsolvable.

To the best of our knowledge, our method is the first exact method except for the trivial brute-force search.
Our method eventually can find strategy-solvable Sudoku clues if exist, and otherwise, our method eventually can recognize strategy-unsolvability.
Besides the clue generation, we presented an application of our method to the problem of determining the minimum number of strategy-solvable Sudoku clues, demonstrating that our method is easily customized with a small modification.

We conducted experiments to compare our method with the generator of Zama and Sasano, using grids varying the positions and the numbers of clues at random.
From the results we observed that our method terminated within $1$ minutes for many grids, showing our method being stable in terms of running time.
However, as the number of clues got closer to $20$, the running time rapidly increased and exceeded the time limit, which is set to $600$ seconds.
On the other hand, the generator of Zama and Sasano often could find solutions much faster even in near $20$ clues, while the performance sharply deteriorated around $45$ clues and exceeded the time limit for all grids with more clue positions.
We also evaluated our method for several instances with $17$ clue positions taken from known minimum Sudokus to see the efficiency for deciding unsolvability.
It remains as future work to improve our method in less clues.

\section*{Acknowledgement}
This work was supported by JSPS KAKENHI Grant Number 17K17725.

\appendix
\begin{landscape}
\begin{table}
\caption{Part 1: The former half of the randomly chosen $30$ minimum Sudokus and the running times (in seconds): each sequence of $9$ digits is separated by a period, and such $9$ sequences in each line represent the rows of a grid.}

\begin{tabular}{cr}\label{tab:min30-1}
Grids & Time (s) \\ \hline
090600000.000080300.000000010.060000800.000205000.000041000.000300702.401000000.500000000 & -\\
050608000.300000070.000000000.000400601.700100500.200000000.061000000.000070020.000090000 & 926\\
000600370.801000000.000200000.070010060.000004500.200080000.060700000.000050800.000000000 & -\\
005000060.000780000.000000000.200000407.001300000.000000800.000601030.040070000.580000000 & -\\
031000000.000400006.000000200.600059000.000010030.400000000.000200800.050000010.700600000 & 1,243\\
000000025.000601000.090000000.805000600.000020000.000000300.040250000.300000790.000800000 & 1,255\\
600500300.000000010.000000000.000000596.010024000.000000000.704000800.000210000.300900000 & -\\
000010600.050000030.000080000.700500020.000002000.008000000.530900000.000400807.000000100 & -\\
000530800.700600000.400000000.100024000.000000630.000000000.050301000.000000042.080000000 & 1,070\\
070060030.500400100.000000000.400501000.300000076.000800000.001000500.060020000.000000000 & -\\
000700380.501000000.000200000.000000506.070400000.000000900.300056000.080010040.000000000 & -\\
603001020.000800500.200000000.050040700.000003000.000200000.040750000.000000031.000000000 & 994\\
000030001.007500000.600000000.810002000.000600350.400000000.003000760.040080000.000000000 & -\\
500300000.000000801.004000600.000600430.710000000.000500000.200000060.000078000.000010000 & 6,564\\
080071000.000040600.000000000.040000008.000600010.200500000.603000500.500200000.000080000 & 1,134\\
\end{tabular}
\end{table}
\end{landscape}

\begin{landscape}
\begin{table}
\caption{Part 2: The latter half of the randomly chosen $30$ minimum Sudokus and the running times (in seconds): each sequence of $9$ digits is separated by a period, and such $9$ sequences in each line represent the rows of a grid.}

\begin{tabular}{cr}\label{tab:min30-2}
Grids & Time (s) \\ \hline
400000076.000081000.000000000.000630004.500000200.017000000.320400000.000000810.000000000 & -\\
020540000.040000006.000000010.080700500.900020000.000006000.603000000.000300200.100000000 & 26,835\\
500060107.030200000.400000000.280000030.000007000.000010000.000800020.000400600.001000000 & -\\
400010000.060000020.000000000.000500270.301400000.008000000.000600100.070002000.100000003 & -\\
008000200.400050000.000600070.000082000.060000050.000300000.950100000.000000306.000000800 & -\\
010000300.000042000.000090000.000800100.205000000.600000004.000310000.900000020.000700050 & 10,090\\
050000200.000700010.600080000.012000050.000600040.000030000.900000308.000001000.000000600 & 852\\
600050043.200000007.000400010.070200000.000060200.010000000.500000800.000730000.000000000 & -\\
304050600.000200000.600000000.080000072.000031000.000000000.120700000.000000340.000000009 & 889\\
040000300.000072000.000010800.200000010.050700000.000050000.000800400.701000000.600300000 & -\\
400080000.000500100.000000200.000034070.001000000.060000000.750000030.000001640.000200000 & 1,021\\
000042500.100000070.000000000.400700000.000000208.000000650.025000000.000830000.060100000 & 869\\
600800000.000090500.000000020.025000700.090000300.000400001.100300008.000050060.000000000 & 849\\
200400000.000000031.000000007.000702500.301000000.900800000.080000400.000030090.070000000 & -\\
007000010.400700000.000800030.200000400.000010000.000300000.000002709.530000000.080000600 & -\\
\end{tabular}
\end{table}
\end{landscape}

\bibliographystyle{plain}
\bibliography{sudoku} 
\end{document}